\documentclass[10pt]{article}

\usepackage{amsthm,mathrsfs}
\usepackage{latexsym}
\usepackage{amssymb}
\usepackage{epsfig}
\usepackage{amsfonts}
\usepackage{color}
\usepackage{amsmath}
\usepackage[T1]{fontenc}
\usepackage{ifthen}
\usepackage{enumerate}
\usepackage{tikz}
\usepackage[margin=3cm]{geometry} 
\usepackage{microtype}
\usepackage{lmodern}
\usepackage{placeins}
\usepackage{wasysym}
\usepackage{stmaryrd}
\usepackage{breqn} 

\newboolean{commentson} 
\setboolean{commentson}{true}

\newcommand{\comment}[1]
{\ifthenelse{\boolean{commentson}}
   {{\par\noindent\mbox{}{\small[ *** #1 ]\par}\noindent\par}}{}}

\newcommand{\PP}{\mathbb P}
\newcommand{\EE}{\mathbb E}
\newcommand{\Var}{{\mathbb V}\!{\rm ar}}

\DeclareMathAlphabet{\pazocal}{OMS}{zplm}{m}{n}

\newcommand{\cc}[1]{\textnormal{#1}} 

\renewcommand{\text}[1]{\mbox{\rm \,#1\,}}        

\renewcommand{\emptyset}{\varnothing}  

\newcommand{\Q}{\mbox{$\mathbb Q$}}

\newcommand{\N}{\mbox{$\mathbb N$}}
\newcommand{\Ordo}[0]{\mathcal{O}}

\newcommand{\opt}[0]{\mbox{\sf Opt}}

\def\Gc{{G^c}}
\def\tdn{{\gamma^t}}
\newtheorem{theorem}{Theorem}[section]
\newtheorem{lemma}[theorem]{Lemma}
\newtheorem{proposition}[theorem]{Proposition}

\newtheorem{corollary}[theorem]{Corollary}
\newtheorem{conjecture}[theorem]{Conjecture}

\newtheorem{problem}{Problem}
\newtheorem{remark}{Remark}

\theoremstyle{remark}
\newtheorem{example}{Example}
\newcommand{\ignore}[1]{}

\title{Tropical Dominating Sets in Vertex-Coloured Graphs}
\author{ J.-A. Angl\`es d'Auriac$^1$, Cs.\ Bujt\'as$^2$,
 A. El Maftouhi$^1$, M. Karpinski$^3$, \\ Y. Manoussakis$^1$, L. Montero$^1$, N. Narayanan$^1$, L. Rosaz$^1$,\\ J. Thapper$^4$, \vspace{1ex} Zs.\ Tuza$^{2,5}$\\
  \small $^1$ Universit\'e Paris-Sud, L.R.I., B\^at. 650, 91405 Orsay Cedex, France.
  \\ \small $^2$ Department of Computer Science and Systems Technology,\\
  \small University of Pannonia, 8200 Veszpr\'em, Egyetem u.~10, Hungary.
 \\ \small
$^3$ University of Bonn, Department of Computer Science,\\ \small
 Friedrich-Ebert-Allee 144, 53113 Bonn, Germany.\\
\small $^4$ Universit\'e Paris-Est, Marne-la-Vall\'ee, LIGM, B\^at. Copernic,\\ \small 5 Bd Descartes, 77454 Marne-la-Vall\'ee Cedex 2, France.
 \\ \small
$^5$ Alfr\'ed R\'enyi Institute of Mathematics,
 Hungarian Academy of Sciences,\\ \small 1053 Budapest, Re\'altanoda u.\ 13--15, Hungary.\\
 \small angles@lri.fr, bujtas@dcs.uni-pannon.hu, hakim.maftouhi@orange.fr,\\
 \small
  marek@cs.uni-bonn.de, yannis@lri.fr, lmontero@lri.fr,\\ \small narayana@gmail.com,
  rosaz@lri.fr, thapper@u-pem.fr, tuza@dcs.uni-pannon.hu}

\date{}

\begin{document}

\maketitle
\begin{abstract}
Given a vertex-coloured graph, a dominating set is said to be tropical
if every colour of the graph appears at least once in the set.
Here, we study minimum tropical dominating sets
 from structural and algorithmic points of view.
First, we prove that the tropical dominating
set problem is NP-complete even when restricted to a simple path.
Then, we establish upper bounds related to various parameters of the graph such as minimum
degree and number of edges. We also give an optimal upper bound for random graphs. Last, we give approximability and inapproximability results for
general and restricted classes of graphs, and establish a FPT algorithm for interval graphs.
\end{abstract}

\vspace*{5mm}

\section{Introduction}
Vertex-coloured graphs are useful in various situations. For instance, the Web graph may be considered as
a vertex-coloured graph where the colour of a vertex represents the content of
the corresponding page (red for mathematics, yellow for physics, etc).
Given a vertex-coloured graph $G^c$, a subgraph $H^c$ (not necessarily induced)
 of $G^c$ is said to be tropical if and only if
each colour of $G^c$ appears at least once in $H^c$. Potentially, any kind of usual structural problems
(paths, cycles, independent and dominating sets, vertex covers, connected components, etc.) could be studied in their tropical version.
This new tropical concept is close to, but quite different from, the colourful concept used for paths in vertex-coloured graphs~\cite{AkLiNi,Li,Lin}.
It is also related to (but again different from) the concept of \emph{colour patterns} used
in bio-informatics~\cite{FeFeHeVi}. Here, we study minimum tropical dominating sets in vertex-coloured graphs.
Some ongoing work on tropical connected components, tropical paths and tropical homomorphisms can be found in~\cite{enpre2,enpre1,enpre3}.
A general overview on the classical dominating set problem can be found in~\cite{Domination}.

Throughout the paper let $G = (V,E)$ denote a simple undirected non-coloured graph. Let $n=|V|$ and $m=|E|$.
Given a set of colours $\mathcal{C}=\{1,...,c\}$, $\Gc= (V^c,E)$ denotes a vertex-coloured graph
where each vertex has precisely one colour from $\mathcal{C}$ and each colour
 of $\mathcal{C}$ appears on at least one vertex.
The colour of a vertex $x$ is denoted by $c(x)$.
A subset $S\subseteq V$ is a \emph{dominating set} of $\Gc$ (or of $G$),
if every vertex either belongs to $S$ or has a neighbour in $S$.
The \emph{domination number} $\gamma(\Gc)$ ($\gamma(G)$) is the
size of a smallest dominating set of $\Gc$ ($G$).
A dominating set $S$ of $\Gc$ is said to be
\emph{tropical} if each of the $c$ colours appears at least once among the
vertices of $S$. The \emph{tropical domination number} $\tdn(\Gc)$ is the size of
a smallest tropical dominating set of $\Gc$.
A \emph{rainbow dominating set} of $\Gc$ is a tropical dominating set with exactly $c$ vertices.
More generally, a $c$-element set with precisely one vertex from each colour
 is said to be a \emph{rainbow set}.
We let $\delta(\Gc)$ (respectively $\Delta(\Gc)$) denote the minimum (maximum) degree of $\Gc$.
When no confusion arises, we write $\gamma$, $\tdn$, $\delta$ and $\Delta$ instead of $\gamma(G)$, $\tdn(\Gc)$, $\delta(\Gc)$ and $\Delta(\Gc)$,
respectively.
We use the standard notation $N(v)$ for the (open) neighbourhood of vertex $v$,
 that is the set of vertices adjacent to $v$, and write $N[v]=N(v)\cup\{v\}$ for its
 closed neighbourhood.
The set and the number of neighbours of $v$ inside a subgraph $H$ is denoted by $N_H(v)$
 and by $d_H(v)$,  independently of whether $v$ is in $H$ or in $V(\Gc) - V(H)$.
Although less standard, we shall also write sometimes $v\in \Gc$  to abbreviate $v\in V(\Gc)$.
\vspace{5pt}

Note that tropical domination in a vertex-coloured graph $G^c$ can also be
 interpreted as ``simultaneous domination'' in two graphs which have a common
 vertex set. One of the two graphs is the non-coloured $G$ itself, the other one is
 the union of $c$ vertex-disjoint cliques each of which corresponds to a colour
 class in $G^c$. The notion of simultaneous dominating set\footnote{Also known under the names
 `factor dominating set' and `global dominating set' in the literature.}
 was introduced by Sampathkumar~\cite{Sampathkumar1989} and independently by Brigham and
 Dutton~\cite{Brigham1990}.
It was investigated recently by Caro and Henning~\cite{Caro2014} and also by further
 authors. Remark that $\delta \ge 1$ is regularly assumed for each factor graph in the
 results of these papers that is not the case in the present manuscript, as we
 do not forbid the presence of one-element colour classes.

The Tropical Dominating Set problem (TDS) is defined as follows.
\begin{problem}TDS\\
Input: A vertex-coloured graph $\Gc$ and an integer $k \geq c$.\\
Question: Is there a tropical dominating set of size at most $k$?
\end{problem}

The Rainbow Dominating Set problem (RDS) is defined as follows.
\begin{problem}RDS\\
Input: A vertex-coloured graph $\Gc$.\\
Question: Is there a rainbow dominating set?
\end{problem}

The paper is organized as follows.
In Section~\ref{npcomplet} we prove that RDS is NP-complete even when graphs are restricted to simple paths.
In Section~\ref{bounds} we give upper bounds for $\tdn$ related to the minimum degree and the number of edges. We give upper
bounds for random graphs in Section~\ref{randomtrop}. In Section~\ref{aprox} we give approximability and inapproximability results for TDS.
We also show that the problem is FPT (fixed-parameter tractable)
 on interval graphs when parametrized by the number of colours.

\section{NP-completeness}\label{npcomplet}
In this section we show that the RDS problem is NP-complete. This implies that the TDS problem is NP-complete too.

\begin{theorem}
The RDS problem is NP-complete, even when the input is restricted to vertex-coloured paths.
\end{theorem}
\begin{proof}
Clearly the RDS problem is in NP.
The reduction is obtained from the 3-SAT problem.
Let $(I,Y)$ be an instance of 3-SAT where $I=(l_1\vee l_2\vee l_3)\wedge(l_4\vee l_5\vee l_6)\wedge\ldots\wedge(l_{X-2}\vee l_{X-1}\vee l_X)$
is a collection of $\tau=X/3$ clauses on a finite set $Y=\{y_1,\ldots,y_m\}$ of boolean variables.
From this instance, we will define a vertex-coloured
path $\mathcal{P}$ such that $\mathcal{P}$ contains a rainbow dominating set if and only if $(I,Y)$ is satisfiable.

In order to define $\mathcal{P}$, we first construct a segment
$\mathcal{P}_{0}=v v' v_0 v_1 \ldots v_{4\tau}$, and we colour
its vertices as follows.
Vertices $v$ and $v'$ are coloured black. Vertices $v_0, v_4, v_8, \ldots, v_{4\tau-4}, v_{4\tau}$ are each coloured with a unique colour.
The remaining vertices, that will be henceforth called \emph{clausal}, are coloured from $v_1$ to $v_{4\tau-1}$ with colours
$1_0, 2_0, 3_0, \ldots, X_0$. Figure~\ref{examplep0} shows $\mathcal{P}_0$ if $X=6$.

\FloatBarrier
\begin{figure}[h]
	\centering
	\includegraphics[scale=.65]{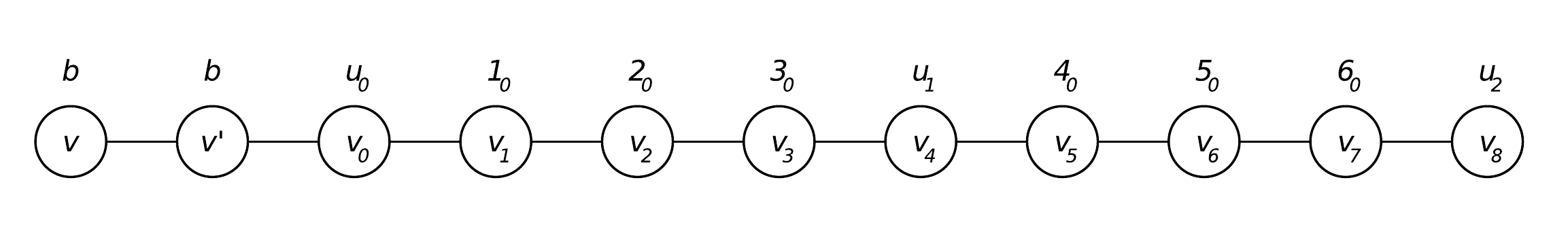}
	\caption{Example of $\mathcal{P}_0$ when $X=6$}
	\label{examplep0}
\end{figure}
\FloatBarrier

Next, we define a number of gadgets as follows.
If a pair of literals $l_i$ and $l_j$ satisfies that $l_i=\overline{l_j}$, we say that $l_j$ is
\emph{antithetic} to $l_i$. For each literal $l_i$, $i=1,\ldots,X$, we consider the list of all literals
$l_{i_1},l_{i_2}, \ldots, l_{i_{k_i}}$ that are antithetic to $l_i$.
Now, to each literal $l_{i_f}$, $f=1,\ldots,k_i$, we associate a \emph{constraint gadget} $w_{i,i_f}$ on five vertices defined as follows.
Vertex $A$ is an \emph{artificial} one and is coloured with a unique colour.
Vertex $P$ is the \emph{positive} one of the gadget, it
corresponds to literal $l_i$ being true and has colour $i_{f}$.
The \emph{middle} vertex $M$ is coloured with black.
Vertex $N$ is the \emph{negative} one, it corresponds to $l_i$ being false and has colour $i_{f-1}$.
Vertex $L$ is the \emph{link} vertex and represents the relation between $l_i$ and $l_{i_f}$.
If $i < i_f$, then $L$ is coloured with colour $c_{i,i_f}$, otherwise, with colour $c_{i_f,i}$. See Figure~\ref{fig:gadget}.
Finally, in order to obtain path $\mathcal{P}$, we concatenate all these gadgets to
$\mathcal{P}_0$ in a serial manner and we close the path with a final vertex $F$ of
a unique colour. Clearly, this construction is polynomial as we have $O(X^2)$ gadgets.

\FloatBarrier
\begin{figure}[h]
	\centering
	\includegraphics[scale=.8]{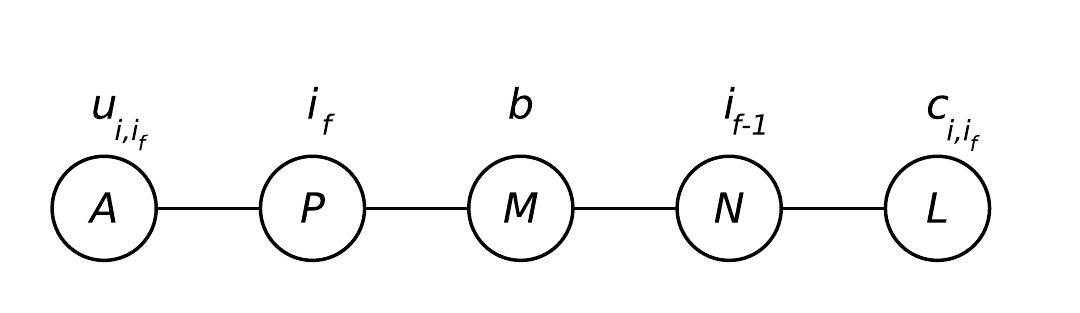}
	\caption{A constraint gagdet $w_{i,i_f}$ (where $i < i_f$)}
	\label{fig:gadget}
\end{figure}
\FloatBarrier

We first prove the \emph{"if case"}.
Consider an assignment to the variables $y_1,\ldots,y_m$ that satisfies the 3-SAT instance. From this assignment we obtain a rainbow
dominating set $\mathcal{D}$ for $\mathcal{P}$ as follows:
\begin{enumerate}
\item We add $v$ and every vertex of a unique colour to $\mathcal{D}$.
\item For each true literal $l_{i}$, we add
the clausal vertex of colour $i_{0}$ to $\mathcal{D}$ and for all $f=1,\ldots,k_i$ (if any), we add the positive vertex $P$ and
the link vertex $L$ of $w_{i,i_f}$ to $\mathcal{D}$.
\item For each false literal $l_{i}$ and all $f=1,\ldots,k_i$ (if any), we add the negative vertex $N$ of $w_{i,i_f}$ to $\mathcal{D}$.
\item If there are vertices with some colour still not present in $\mathcal{D}$, we add them to the set.
\end{enumerate}
We can check that each colour is present exactly once in $\mathcal{D}$.
In fact, this conclusion is straightforward for black colour,
for every unique colour and for each colour $i_f$.
For the colour on the link vertex $L$ of a gadget $w_{i,j}$, as literals $l_{i}$ and $l_{j}$ are antithetic,
then exactly one of them is true, so it stands that either $c_{i,j}$ is present once in $\mathcal{D}$ if $i<j$, or $c_{j,i}$ otherwise.

We show now that $\mathcal{D}$ is a dominating set.
Observe first that $v'$ is dominated by $v$ (and also by $v_0$).
Then, as every clause contains at least one true literal, every
clausal vertex is either in $\mathcal{D}$ or has some neighbour in $\mathcal{D}$.
Last, every constraint gadget has either its negative vertex or both
its positive and link vertex in $\mathcal{D}$. In both cases, all vertices of the
gadget are covered by $\mathcal{D}$.
Therefore $\mathcal{D}$ is a rainbow dominating set as claimed.
An example of the construction of $\mathcal{P}$ and its corresponding rainbow dominating set $\mathcal{D}$ for an arbitrary assignment to the
variables is shown in Figure~\ref{fullpath}.

\FloatBarrier
\begin{figure}[h]
	\centering
	\includegraphics[scale=.65]{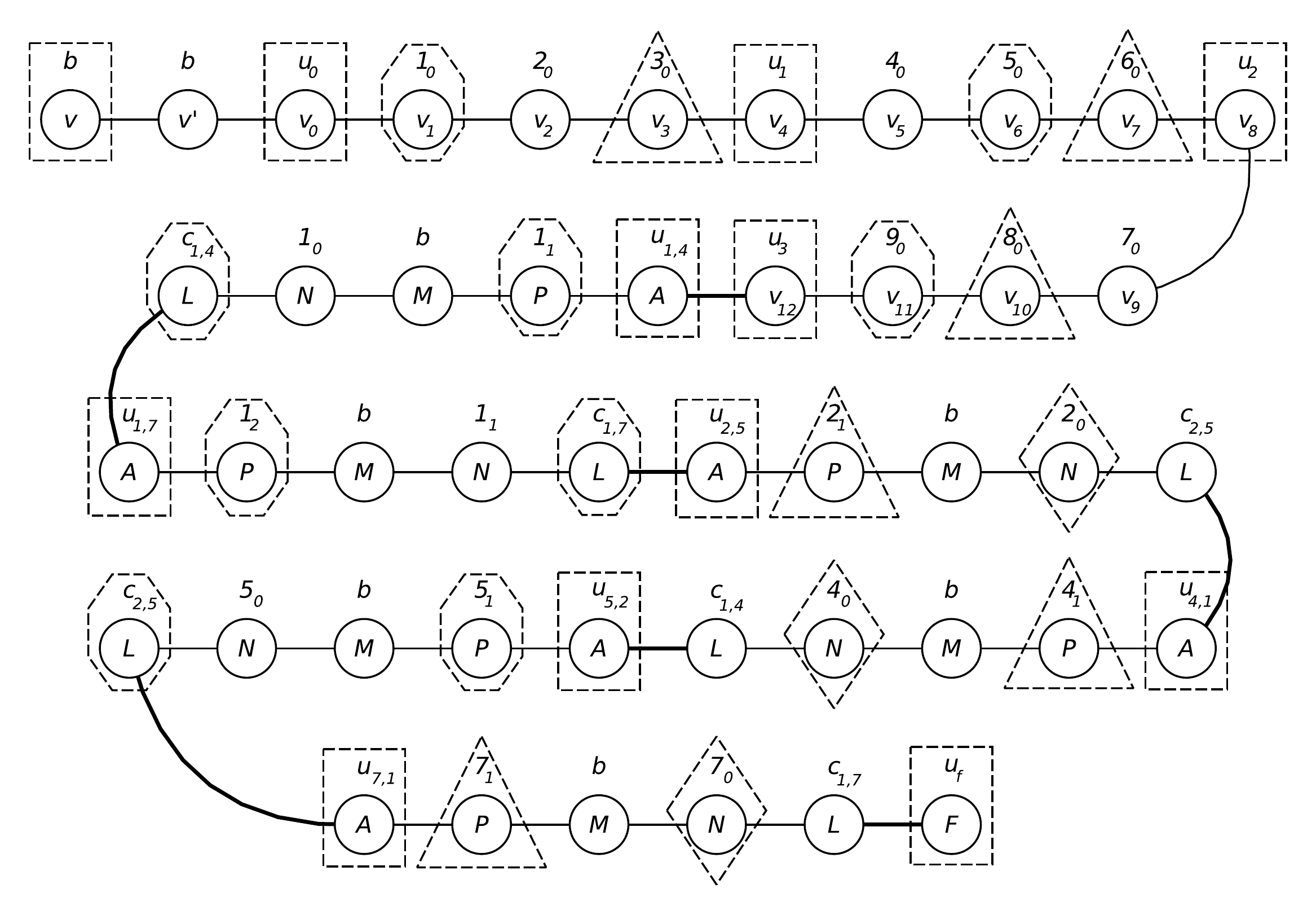}
\caption{Construction of $\mathcal{P}$ for the formula $(y_1\vee \overline{y_2}\vee y_3)\wedge(\overline{y_1}\vee y_2\vee y_3)\wedge(\overline{y_1}\vee y_3\vee y_4)$
where $\mathcal{P}=\mathcal{P_0}w_{1,4}w_{1,7}w_{2,5}w_{4,1}w_{5,2}w_{7,1}F$. The colour of the vertices is shown on top of them and the thick
edges represent the division of the gadgets. Vertices surrounded by dotted $\APLbox$,$\octagon$,$\lozenge$ and $\APLup$ are the ones taken by the steps $1,2,3$ and $4$, respectively, to obtain the
rainbow dominating set $\mathcal{D}$ corresponding to the assignment $y_1=y_2=y_4=True$ and $y_3=False$. Note that the step $4$ is not needed
for the set to be dominating but it is to become rainbow.}
	\label{fullpath}
\end{figure}
\FloatBarrier

We now prove the \emph{"only if"} case.
Given the path $\mathcal{P}$ constructed as before, let $\mathcal{D}$ be a rainbow dominating set for $\mathcal{P}$.
We consider first a partial assignment where for every clausal vertex of colour $i_0$ that is in
$\mathcal{D}$, we assign the value to the corresponding variable such that $l_i$ is true.
Suppose by contradiction that this assignment method leads to some incoherences, that is,
some variable ends up being assigned both true and false.
It implies that $\mathcal{D}$ contains two clausal vertices of colour $i_0$ and
$j_0$, respectively, such that $l_i$ and $l_j$ are antithetic. Suppose without loss of generality that $i < j$.
As $l_i$ and $l_j$ are antithetic, there exist two gadgets $w_{i,i_f}$, $i_f=j$,
and $w_{j,j_{f'}}$, $j_{f'}=i$, where the link vertices are both coloured $c_{i,j}$ (as $i<j$).
We will show that both of those vertices are in $\mathcal{D}$ lead to a contradiction.
Consider $w_{i,i_f}$ (a similar proof works for $w_{j,j_{f'}}$).

Suppose first that $f=1$. If $L$ does not belong to $\mathcal{D}$, then
$N$ must be in $\mathcal{D}$ as $M$ being black, cannot belong to $\mathcal{D}$ since $\mathcal{D}$ must already contain either $v$ or $v'$ that are
also coloured black. This is a contradiction since the colour of $N$ is $i_0$ and $\mathcal{D}$ already contains the clausal vertex on that colour.
Therefore $L$ belongs to $\mathcal{D}$.

Suppose next that $f>1$. By the same argument as before, if $L$ does not belong to $\mathcal{D}$ then $N$ must be in $\mathcal{D}$.
As the colour of $N$ is $i_{f-1}$, then the vertex $P$ of $w_{i,i_{f-1}}$ cannot belong to $\mathcal{D}$ since it has the same colour.
Therefore, as $\mathcal{D}$ must dominate the vertex $M$ of $w_{i,i_{f-1}}$, we have that the vertex $N$ on colour $i_{f-2}$ of
$w_{i,i_{f-1}}$ belongs to $\mathcal{D}$. Following the same argument, we have that the negative vertex $N$ belongs to $\mathcal{D}$
for each $w_{i,i_{k}}$, $k=f,\ldots,1$. This is a contradiction since the vertex $N$ of $w_{i,i_1}$ has colour $i_0$
and $\mathcal{D}$ already contains the clausal vertex on that colour. Therefore $L$ belongs to $\mathcal{D}$.

As the same reasoning holds for $w_{j,j_{f'}}$ as well, $\mathcal{D}$
contains two vertices of colour $c_{i,j}$. This is a contradiction to the hypothesis that $\mathcal{D}$
is rainbow. Hence our partial assignment is coherent.
We complete the assignment by setting to true every unassigned variable.

Now, it remains to see that the obtained assignment satisfies the 3-SAT instance.
Indeed, as every clausal vertex is covered by $\mathcal{D}$,
then for every clause we have one literal who was assigned to true, that is,
the assignment is a solution to the 3-SAT instance.
This completes the argument and the proof of the theorem.
\end{proof}

\section{Upper Bounds}\label{bounds}

We begin with an easy observation which intends to avoid some
trivial technical distinctions later on.

\begin{proposition}
 If $\Gc$ is a vertex-coloured graph with $c$ colours on $n$ vertices,
  and $\delta(\Gc)\ge n-c$, then every rainbow set is a rainbow dominating set of $\Gc$.
 As a consequence, $\tdn(\Gc)=c$ holds.
\end{proposition}

\begin{proof}
Let $D$ be any rainbow set, and $v\in V(G^c)-D$ any vertex. Then
 $|N[v]|\ge\delta(\Gc)+1\ge n-c+1=|V(G^c)-D|+1$, thus $D$ dominates $v$.
\end{proof}

Without further reference to this proposition, throughout the text below,
 we shall disregard whether or not any proof works for graphs of
 minimum degree at least $n-c$, because we know the tropical domination
 number of those graphs exactly.

\begin{proposition}\label{thm:first} For any graph $\Gc$, $\tdn \le
\gamma+c-1$.  Furthermore, there are extremal graphs that attain this bound.
\end{proposition}

\begin{proof}
The upper bound follows from the fact that taking a minimal dominating set of $\Gc$ and then adding vertices from every colour
not present in the dominating set gives a tropical dominating set. To construct extremal graphs, for $\gamma\in \N$,
consider a cycle of length $3\times \gamma$ and add $c-1$ leaves to
one vertex $u$ of the cycle. Colour the new leaves with
unique colours and the rest of the graph with a single colour.
Taking $u$ and every vertex which its distance to $u$
is a multiple of three, is a minimum tropical dominating set that attains the bound.
The reader can easily check that $\tdn = \gamma+c-1$.
\end{proof}

\begin{proposition} Given an integer $k>0$, if $m \geq {n-k+c-1 \choose 2} + n-k$ and $n> k+c-2$, then $\tdn\leq k$.
Furthermore, there are graphs with that number of edges and $\tdn = k$.
\end{proposition}
\begin{proof}
We can check that $m \geq {n-k+c-1 \choose 2} + n-k > \frac{1}{2}(n-k+c)(n-k+c-2)$ for $n> k+c-2$, therefore,
by the contrapositive of Vizing's theorem stated in~\cite{Vizing}, we obtain a minimum dominating set on
at most $k-c+1$ vertices. Then we may add at most $c-1$ other vertices to this dominating set
to represent the colours that are absent. Thus obtain a tropical one.

To construct a graph that attains this bound, do the following. Consider a
clique on $n-k+c-1$ vertices. Pick some $c-1$ vertices and colour them with
$c-1$ unique colours. Every other vertex of the graph is coloured with the
remaining colour. Let $A$ be the set of the remaining $n-k$ vertices. Add a
new set $B$ of $k-c+1$ vertices.  Now make a bipartite graph between these
two sets such that each vertex in the set $A$ is adjacent to exactly one vertex of $B$.
It is easy to check that this graph has
the necessary number of edges. Also, at least $k-c+1$ vertices are needed to
dominate the vertices in $B$. This set can represent only one colour. So we need
to add the $c-1$ uniquely coloured vertices to the set to get the required
tropical dominating set.
\end{proof}

\def\Gd{{G_{\delta}}}

\begin{conjecture}\label{delta-general}
Let $\Gc$ be a connected graph with minimum degree $\delta$
and $c < n$. Then, $\tdn \leq \frac{(n-c+1) \delta}{3\delta-1} + c-1$.
\end{conjecture}

We are motivated to raise this conjecture by the fact that its
 particular case for $c=1$ (i.e., simple graphs without colours) holds true.
Its proof is a long story, however, taking nearly a half century, along the
 works by Ore~\cite{Ore-1962} (1962, $\delta=1$),
  Blank~\cite{Blank:1973} (1973), and independently McCuaig and Shepherd~\cite{McCuaig1989} (1989) ($\delta=2$),
  Arnautov~\cite{arnautov-1974} (1974, $\delta\geq 6$ with a stronger upper bound in general),
  Reed~\cite{reed-1996} (1996, $\delta=3$),
  Xing, Sun and Chen~\cite{XSC-2006} (2006, $\delta=5$)
  and Sohn and Xudong~\cite{SohnXudong:09} (2009, $\delta=4$).


\begin{lemma}\label{delta-general-uncol}
 If $G$ is a connected graph with $n$ vertices and minimum degree $\delta$,
  then $\gamma(G)\le\frac{\delta n}{3\delta-1}$,
  with precisely seven exceptions if $\delta=2$.
\end{lemma}

We first prove Conjecture \ref{delta-general} for $\delta = 1$.

\begin{theorem}
If $\Gc$ is a connected graph with $n>c\ge 1$, then  $\tdn \leq \frac{n+c-1}{2}$.
\end{theorem}

\begin{proof}
Let $\Gc$ be a connected graph with minimum degree $\delta \geq 1$.
Let $A$ be a subset of vertices of $\Gc$ with each of the $c$ colours once.
Let $B = \{ v \in \Gc-A$ : $v$ has a neighbour in $A \}$.
Clearly $B \neq \emptyset$ as $\Gc$ is connected and has $n>c$.
Consider the graph $\Gc-A-B$. If this graph is empty, then $A$ is a tropical dominating set for $\Gc$ of size $c$
and we are done.

Now, suppose first that $\Gc-A-B$ has no isolated vertices. Then, since the domination number of a graph without isolated vercites is at most
half the order (cf.~\cite{Domination}),
we obtain a dominating set for $\Gc-A-B$ of size at most $\frac{n-c-|B|}{2}$. Now, as this is less than or equal to $\frac{n-c-1}{2}$,
adding the $c$ vertices of $A$ to this dominating set we obtain a tropical one for $\Gc$ of size at most $\frac{n-c-1}{2}+c=\frac{n+c-1}{2}$ as desired.

Suppose next that $\Gc-A-B$ has isolated vertices. Let $D$ be the set of isolated vertices of $\Gc-A-B$.
Let $k_1=|B|$ and $k_2=|D|$. As above, we obtain a dominating set $S$ for $\Gc-A-B-D$ of size $\frac{n-c-k_1-k_2}{2}$.
Now, if $k_1 > k_2$, then $S \cup D \cup A$ is a tropical dominating set of size $\frac{n-c-k_1-k_2}{2}+k_2+c \leq \frac{n+c-1}{2}$.
Otherwise, if $k_1 \leq k_2$, let $v \in A$ be a vertex such that its colour appears on some vertex in $B$.
Then, $S \cup B \cup (A-\{v\})$ is a tropical dominating set for $\Gc$ as $v$ is dominated by a vertex either in $A$ or in $B$.
The size of this tropical dominating set is $\frac{n-c-k_1-k_2}{2}+k_1+c-1 < \frac{n+c-1}{2}$ and this completes the proof.
\end{proof}

The validity of Conjecture~\ref{delta-general},
 for $\delta \ge 9$, will be a consequence of the following theorem.

\begin{theorem}\label{delta-9}
For any connected graph $G^c$  with minimum degree $\delta(G^c)=\delta$, either $\gamma^t (G^c)=c$ or
$\gamma^t (G^c) < \frac{1+\ln(\delta+1)}{\delta+1}(n-c+1)+c-1$ holds.
\end{theorem}

\begin{proof}
 Given a graph $G^c=(V^c,E)$ satisfying the conditions, first we pick one vertex
from each colour class. If this set $A$ dominates all vertices, then
$\gamma^t=c$. Otherwise, there is a vertex $v$ which is undominated, that is $N[v]\cap A =
  \emptyset$; and further, there is a vertex $v^\prime \in A$ with $c(v')=c(v)$. Then,
the set $D_p=(A - \{v'\})\cup \{v\}$ is rainbow and dominates at
least $|N(v)|\ge \delta$ vertices from $V^c - D_p$. Observe that in
this case also $n\ge c+\delta+1$ and $\delta \ge 1$ hold.

From now on, we  apply the probabilistic method in a similar way as it is done on
 pages 4--5 of \cite{AlonSpencer} concerning $\gamma(G)$.
Let $D_0$ be a subset of $V^c - D_p$ chosen at random, such that
 independently for each $u \in V^c - D_p$,
  $${\mathbb{P}}(u \in D_0)= \frac{\ln(\delta+1)}{\delta+1}.$$
The cardinality of $D_0$ is a random variable, which is the sum of the
  random variables $\xi_u$ defined as
 $\xi_u=1$ if $u\in D_0$ and $\xi_u=0$ if $u\notin D_0$.
Therefore, the expected number of selected vertices is
$${\mathbb{E}}(|D_0|)=\frac{\ln(\delta+1)}{\delta+1}(n-c).$$
The set $D_p\cup D_0$ dominates all the at least $c+\delta$ vertices
in $D_p \cup N(v)$. Consider the set $D_1$ consisting of those
vertices which are not dominated by $D_p \cup D_0$.
 For each vertex $u\in V - (D_p\cup N(v))$, we have
$${\mathbb{P}}(u \in D_1)\le
\left(1-\frac{\ln(\delta+1)}{\delta+1}\right)^{\delta+1}< e^{-\ln
(\delta +1)}=\frac{1}{\delta +1} $$
  that implies $${\mathbb{E}}(|D_1|) < \frac{1}{\delta+1}(n-c-\delta).$$

Clearly, $D_p \cup D_0 \cup D_1$ is a tropical dominating set of $G$
and $\gamma^t(G)$ is not greater than the expected value of its
size. Therefore, we have
\begin{eqnarray*}
 \gamma^t(G) &<& c +\frac{\ln(\delta+1)}{\delta+1}(n-c)+
\frac{1}{\delta+1}(n-c-\delta)\\
              &=& \frac{1+\ln(\delta+1)}{\delta+1}(n-c+1)+c-1 +
              \frac{1}{\delta+1}-\frac{1+\ln(\delta+1)}{\delta+1}\\
              &< & \frac{1+\ln(\delta+1)}{\delta+1}(n-c+1)+c-1
\end{eqnarray*}
 which completes the proof.
\end{proof}

\begin{corollary}
For any connected graph $G^c$  with minimum degree  $\delta \ge 9$,
either $\gamma^t (G^c)=c$ or $\gamma^t (G^c)< \frac{\delta(n-c+1)}{3\delta-1}+c-1$
holds.
\end{corollary}
\begin{proof} For every $\delta \ge 9$,
 $\frac{1+\ln(\delta+1)}{\delta+1} < \frac{\delta}{3\delta-1}$
 holds. Then,  Theorem~\ref{delta-9} implies
  $$\gamma^t (G^c) < \frac{1+\ln(\delta+1)}{\delta+1}(n-c+1)+c-1
 <  \frac{\delta}{3\delta-1}(n-c+1)+c-1,$$
 if $\gamma^t (G^c)\neq c$.
\end{proof}

For $2\leq \delta \leq 8$ we present the following result that is slightly weaker than Conjecture~\ref{delta-general}.

\begin{theorem}
Let $\Gc$ be a connected graph with minimum degree $\delta$ such that $2\leq \delta \leq 8$
and $c < n$. Then, $\tdn \leq \frac{(n-c) \delta}{3\delta-1} + c + \frac{\delta(\delta-2)}{3\delta-1}$.
\end{theorem}

\begin{proof}
Let $A$ be a subset of vertices of $\Gc$ with each of the $c$
colours once. Let $B = \{ v \in \Gc-A$~: $d_{\Gc-A}(v) < \delta \}$.
 Clearly, if $B = \emptyset$, then by Lemma~\ref{delta-general-uncol}
  we obtain a dominating set for
$\Gc-A$ of size $\frac{(n-c) \delta}{3\delta-1}$. Thus, adding $A$
to this set we obtain a tropical dominating set for $\Gc$ of size
$\frac{(n-c) \delta}{3\delta-1} + c \leq \frac{(n-c)
\delta}{3\delta-1} + c + \frac{\delta(\delta-2)}{3\delta-1}$. Assume
therefore that $B \neq \emptyset$ and let $|B|=b$.
 We intend to keep $A$ in the tropical dominating set to be constructed,
 by extending $A$ with a subset of $\Gc-A$ that dominates $\Gc-A-B$.

First we can assume that for every vertex $v \in B$, $d_{\Gc-A-B}(v)
\geq 2$. Otherwise, if there is a vertex $v \in B$ with a unique
neighbour $w \in \Gc-A-B$, we can set $N_{\Gc-A}[v] = N_{\Gc-A}[w]$
to make the degree of $v$ at least $\delta$ and therefore any
dominating set extending $A$
 that contains $v$ is equivalent to one containing $w$
instead.

Second, if $b \geq \delta-1$, then making $B$ a complete graph we obtain that for every vertex $v \in B$,
$d_{\Gc-A}(v) = d_{\Gc-A-B}(v) + d_{B}(v) \geq 2 + \delta - 2 = \delta$.
 We can therefore apply Lemma~\ref{delta-general-uncol}
to obtain a dominating set $S$ for $\Gc-A$ of size $\frac{(n-c) \delta}{3\delta-1}$
except for the case when $\delta = 2$ and some components of $\Gc-A$ are one of the seven exceptions
 (which are also listed in~\cite{Domination}).
In this special case, if at least one of the vertices of a component is dominated
from outside (as is the case in our problem since every vertex in $B$ is dominated by some vertex in $A$), then each of them satisfies the bound.
Therefore we have that $S \cup A$ is a tropical dominating set for $\Gc$ of size
$\frac{(n-c) \delta}{3\delta-1} + c \leq \frac{(n-c) \delta}{3\delta-1} + c + \frac{\delta(\delta-2)}{3\delta-1}$.

We can conclude that $b \leq \delta-2$, and consequently
 $\delta\ge 3$ holds.
Consider now the following graph obtained from $\Gc-A$ plus a complete graph $C$ on $\delta - b -1$ new vertices.
First make $B$ a complete graph. Join every vertex in $B$ to every vertex in $C$. Finally, for every vertex $w \in C$ set
$N_{\Gc-A-B}(w)=N_{\Gc-A-B}(v)$ for some vertex $v \in B$. Clearly this new graph has minimum degree $\delta$ therefore
by~\cite{Domination} we obtain a dominating set $S$ of size at most
 $\frac{(n-c+\delta-b-1) \delta}{3\delta-1}$.
Now we obtain a tropical dominating set for the original graph $\Gc$ as follows. If any of the vertices of $C$ belongs to $S$ we just
delete them from $S$ and add instead the chosen vertex $v \in B$. Then add $A$ to the dominating set.
This new set is dominating as every vertex in $B$ is dominated by some vertex in $A$ in $\Gc$ and clearly it is tropical.
Finally, its size is not greater than
 $\frac{(n-c+\delta-b-1) \delta}{3\delta-1} + c$ and this number is maximum when $b$ is as small as possible, that is, $b=1$.
We obtain then that the size of the tropical dominating set is $\frac{(n-c+\delta-2) \delta}{3\delta-1} + c = \frac{(n-c) \delta}{3\delta-1} + c + \frac{\delta(\delta-2)}{3\delta-1}$
which completes the proof.
\end{proof}

\begin{proposition} Let $\Gc$ be super-dense, i.e., $\delta(\Gc)> (n-1) - \sqrt{n-1}$.
Then $\tdn\le c+1$.
\end{proposition}

\begin{proof}
Since $\delta(\Gc)> (n-1) - \sqrt{n-1}$, then the
complement $\overline{\Gc}$ of $\Gc$ satisfies $\Delta(\overline{\Gc})<\sqrt{n-1}$.
Therefore, the diameter of $\overline{\Gc}$ is
at least $3$.  Thus according to a theorem from~\cite{Domination}, the domination number $\gamma$ of
$\Gc$ is at most $2$. Now, the result for $\tdn$ follows from Proposition~\ref{thm:first}.
\end{proof}

\section{Tropical dominating sets in random graphs}\label{randomtrop}

In this section we study the tropical domination parameter of a \emph{randomly vertex-colored random graph}.
Recall that the random graph $G(n,p)$ is the graph on $n$ vertices where each of the possible $\binom{n}{2}$ 
edges appears with probability $p$, independently. For more details on random graph theory, 
we refer the reader to \cite{bollobas:2001} and \cite{jlr:2000}. Given a positive integer $c$, 
let $G(n,p,c)$ be the vertex-colored graph obtained from $G(n,p)$ by coloring each vertex with one of the colors 
$1, 2,\dots,c$ uniformly and independently at random. The choice of colors is independent of the existence of edges. 
In what follows, we will say that $G(n,p,c)$ has a property $\pazocal{Q}$ \textit{ asymptotically almost surely} 
(abbreviated a.a.s.) if the probability it satisfies $\pazocal{Q}$ tends to $1$ as $n$ tends to infinity.
For convenience, we will use the notation $b=1/(1-p)$.

The domination number of the random graph $G(n,p)$ has been well studied, see for example 
\cite{dryer:2000}, \cite{elmaftouhi:1993} and \cite{wieland:2001}. In Particular, Wieland and 
Godbole \cite{wieland:2001} proved the following two-point concentration result.

\begin{theorem}{\label{th:1} \rm (\cite{wieland:2001})~~}
Let $p = p(n)$ be such that $10 \sqrt{(\log\log n)/\log n }\leq p < 1$. Set $b=1/(1-p)$. Then a.a.s. 
the domination number of the random graph $G(n,p)$ is equal to 
$$
\big\lfloor\log_b n-\log_b\left[(\log_b n)(\log n)\right]\big\rfloor+1 ~~~\textit{or}~~~
\big\lfloor\log_b n-\log_b\left[(\log_b n)(\log n)\right]\big\rfloor+2.
$$
\end{theorem}

Recently Glebov, Liebenau and Szab\'o \cite{glebov} strengthened this two-point concentration result by 
extending the range of $p$ down to $(\log^2 n)/\sqrt n$ .

We are interested here in the maximum number of colors $c$ that can be used so that a.a.s. $G(n,p,c)$
has a tropical minimal dominating set. We only deal with the case when $p$ is fixed. It follows from
Theorem \ref{th:1} that the number of colors should not exceed 
$\big\lfloor\log_b n-\log_b\left[(\log_b n)(\log n)\right]\big\rfloor+2$. We show in Theorem \ref{th:2} that 
this upper bound is achieved. This result can be extended to hold 
when $p=p(n)$ tends to $0$ sufficiently slowly. This could be the subject of another study since, in this case,
the proof is very technical. 

For $c\geq 1$, let $X_c$ be the random variable counting the number of tropical dominating sets of size $c$ in $G(n,p,c)$.
\begin{equation*}
 X_c=\sum_{j=1}^{\binom{n}{c}}I_j,
\end{equation*}
where $I_j$ is the indicator random variable indicating if the \textit{j-th} $c$-set is both tropical and 
dominating in $G(n,p,c)$.

In order to prove Theorem \ref{th:2}, we need the following lemma.

\begin{lemma}
Let $0<p<1$ be fixed. Set $b=1/(1-p)$. Denote by $X_c$ the number of tropical dominating sets of size $c$ in $G(n,p,c)$.
Then
\begin{equation*}
\EE(X_c)\to\infty~~~\textit{as}~~~n\to\infty,
\end{equation*}
for $c=c(n)=\big\lfloor\log_b n-\log_b\left[(\log_b n)(\log n)\right]\big\rfloor+2.$
\end{lemma}\label{lem:1}
\begin{proof}
Clearly, for $1\leq j\leq \binom{n}{c}$, we have
\begin{equation*}
 \EE(I_j)=\PP\left[ I_j=1\right]=\left(1-(1-p)^c\right)^{n-c}\frac{c!}{c^c},
\end{equation*}
where $\left(1-(1-p)^c\right)^{n-c}$ is the probability that a given $c$-set $A_j$ is dominating, and 
$c!/c^c$ is the probability that $A_j$ is tropical.  By the linearity of expectation, we have
\begin{eqnarray}
 \EE(X_c)&=&\binom{n}{c}\left(1-(1-p)^c\right)^{n-c}\frac{c!}{c^c}.\nonumber\\
	 &=&\frac{(n)_c}{c^c} \left(1-(1-p)^c\right)^{n-c}.\nonumber\\
	 &\geq&\frac{(n)_c}{c^c} \left(1-(1-p)^c\right)^{n}.\nonumber
\end{eqnarray}
Using the inequality $1-x\geq -x/(1-x)$ and the estimate $(n)_c=\left(1+o(1)\right)n^c$, we have
\begin{eqnarray}
 \EE(X_c) &\geq& (1-o(1))\frac{n^c}{c^c} \exp\left[\frac{-n(1-p)^c}{1-(1-p)^c}\right]\nonumber\\
	  &=& (1-o(1))\exp\left[\Psi(c)\right]\nonumber,
\end{eqnarray}
where
\begin{equation*}
 	  \Psi(c)=c\log n-c\log c - \frac{n(1-p)^c}{1-(1-p)^c}.
\end{equation*}
Recall that $c=\big\lfloor\log_b n-\log_b\left[(\log_b n)(\log n)\right]\big\rfloor+2$. Thus,
\begin{equation*}
(1-p)^{c}\leq (1-p)\log_bn\log n
\end{equation*}
and
\begin{equation*}
\frac{n(1-p)^c}{1-(1-p)^c}\leq (1-p)\log_b n\log n +o(1).
\end{equation*}
It follows that
\begin{equation*}
\Psi(c) \geq c\log n -c\log c - (1-p)\log_b n\log n +o(1).
\end{equation*}
Straightforward calculations show that
\begin{equation*}
 \Psi(c)\geq (p-o(1))\log_b n\log n.
\end{equation*}
Since $p$ is fixed, we conclude that $\psi(c)$, and thus also $\EE(X_c)$, tends to infinity as $n\to\infty$.
\end{proof}

\begin{theorem}\label{th:2}
 Let $0<p<1$ be fixed and set $b=1/(1-p)$. Let $c=c(n)$ be the function defined  by
$$
c(n)=\big\lfloor\log_b n-\log_b\left[(\log_b n)(\log n)\right]\big\rfloor+2.
$$
Then a.a.s. $G(n,p,c)$ contains a tropical dominating set of size $c$.
\end{theorem}
\begin{proof}
To prove the Theorem, we use the second moment method. 
For this, we need to estimate the variance $\Var(X_c)$ of the number of tropical dominating 
sets of size $c$.
\begin{eqnarray}
\Var(X_c)&=&\sum_{i=1}^{\binom{n}{c}}\sum_{j=1}^{\binom{n}{c}}\EE(I_iI_j)-\EE^2(X_c)\nonumber\\
	 &=&\sum_{k=0}^{c}\binom{n}{c}\binom{c}{k}\binom{n-c}{c-k}\EE(I_1I_k)-\EE^2(X_c)\nonumber,\\
	 &=&\EE(X_c) + \sum_{k=0}^{c-1}\binom{n}{c}\binom{c}{k}\binom{n-c}{c-k}\EE(I_1I_k)-\EE^2(X_c),\label{eq:1}
\end{eqnarray}
where $I_k$ is the indicator random variable of any generic $c$-set $A_k$ that intersect the first $c$-set $A_1$ 
in $k$ elements. We have
\begin{eqnarray}
\EE(I_1I_k )&=&   \PP\big[~A_1\textit{~dominates~and~} A_k\textit{~dominates~}\big]\times
		  \PP\big[~A_1\textit{~and~}A_k\textit{~are tropical~}\big]\nonumber\\
	   &\leq& \PP\big[~A_1\textit{~dominates~}\overline{(A_1\cup A_k)}~\wedge~ 
		  A_k\textit{~dominates~}\overline{(A_1\cup A_k)}~\big]\times
		  \PP\big[~A_1\textit{~and~}A_k\textit{~are tropical~}\big]\nonumber\\
	   &=&   \left(1-2(1-p)^k+(1-p)^{2c-k}\right)^{n-2c+k}\times
		  \frac{c!(c-k)!}{c^{2c-k}}\label{eq:2}		  
\end{eqnarray}
Relations (\ref{eq:1}) and (\ref{eq:2}) imply
\begin{equation}\label{eq:3}
 \Var(X_c) \leq \EE(X_c)+\EE^2(X_c)A+B,
\end{equation}
where
\begin{equation*}
A=\binom{n}{c}^{-1}\binom{n-c}{c}\Big(1-(1-p)^c\Big)^{-2c}-1
\end{equation*}
and
\begin{equation*}
B =	  \binom{n}{c}\sum_{k=1}^{c-1}\binom{c}{k}\binom{n-c}{c-k}
	  \frac{c!(c-k)!}{c^{2c-k}}
	  \Big(1-2(1-p)^k+(1-p)^{2c-k}\Big)^{n-2c+k}.
\end{equation*}
Using once again the inequality $1-x\geq -x/(1-x)$, we can bound $A$ as follows
\begin{eqnarray}
A&\leq& e^{\frac{-c^2}{n}}\exp\left[\frac{2c(1-p)^c}{1-(1-p)^c}\right]-1\nonumber\\
&\leq& \exp\left[\frac{-c^2}{n}+2c(1-p)^c(1+o(1))\right]-1.\nonumber
\end{eqnarray}
Thus, for $c=\big\lfloor\log_b n-\log_b\left[(\log_b n)(\log n)\right]\big\rfloor+2$
\begin{eqnarray}
  A &\leq& \left(2c(1-p)^c-\frac{c^2}{n}\right)(1+o(1))\nonumber\\
    &=& O\left(\frac{(\log n)^3}{n}\right).\label{eq:4}
\end{eqnarray}
\medskip
\noindent Now, we need to estimate $B$. We have
\begin{equation*}
   B\leq  \binom{n}{c}\sum_{k=1}^{c-1}g(k)\nonumber,
\end{equation*}
where
\begin{equation*}
g(k)=\frac{(c!)^2}{k!(c-k)!}\frac{n^{c-k}}{c^{2c-k}}\exp\Big[(n-2c)\left((1-p)^{2c-k}-2(1-p)^{c}\right)\Big].
\end{equation*}
We shall show that 
\begin{itemize}
 \item[(i)] there exists $k_0=k_0(n)\to \infty$ such that $g$ is decreasing if $k\leq k_0$ and increasing if 
$k\geq k_0$,
\item[(ii)] $g(1)\geq g(c-1)$,
\end{itemize}
which will imply that 
\begin{equation}\label{eq:5}
\sum_{k=1}^{c-1}g(k)\leq cg(1).
\end{equation}
Clearly,
\begin{equation*}
 g(1)\geq g(c-1)
\end{equation*}
iff
\begin{equation*}
\frac{n^{c-2}}{c^{c-2}}\exp\Big[(n-2c)\Big((1-p)^{2c-1}-(1-p)^{c+1}\Big)\Big]\geq 1
\end{equation*}
iff
\begin{equation*}
\exp\Big[(c-2)\log n-(c-2)\log c +\frac{(n-2c)}{n^2b^{3}}(\log_b n\log n)^2
-\frac{(n-2c)}{nb^{3}}\log_b n\log n\Big]
\geq 1
\end{equation*}
iff
\begin{equation*}
\exp\Big[\left( 1-\frac{1}{b^{3}}+o(1)\right)\log_b n\log n\Big]
\geq 1.
\end{equation*}
Since the left-hand side of the last inequality tends to infinity as $n\to \infty$, the above condition is thus 
satisfied and (ii) is proved.

\medskip
\noindent For $1\leq k\leq c-1$, let
\begin{equation*}
 h(k)=\frac{g(k+1)}{g(k)}=\frac{c(c-k)}{n(k+1)}\exp\Big[(n-c)p(1-p)^{2c-k-1}\Big].
\end{equation*}
It is straightforward to see that, for 
$c=\big\lfloor\log_b n-\log_b\left[(\log_b n)(\log n)\right]\big\rfloor+2$, 
\begin{equation*}
h(k)\geq 1 
\end{equation*}
iff 
\begin{equation*}
(1-p)^{k-3}\log \left[\frac{n(k+1)}{c(c-1)}\right]\leq\frac{(n-c)}{n^2}p\big(\log_b n\log n\big)^2
\end{equation*}
iff 
\begin{equation*}
(1-p)^{k-3}\log n \left(1+\delta(k)\right)\leq\frac{(n-c)}{n^2}p\big(\log_b n\log n\big)^2,
\end{equation*}
where $\delta(k)=\log \big[n(k+1)/c(c-1)\big]/\log n=O(\log c/\log n)$.
Therefore $h(k)\geq 1$ if and only if
\begin{equation*}
k\geq\log_b n-2\log_b\left(1-\frac{c}{n}\right)-2\log_b\log_b n - \log_b\log n -\log_b p
+\log_b\big(1+\delta(k)\big)+3.
\end{equation*}
The right-hand side of the above inequality is of the form $a_n+o(1)$, $a_n\to \infty$. Thus, there exists 
$k_0=k_0(n)$ such that $h(k)\geq1$ if and only if $k\geq k_0$. We have thus shown (i).
Combining (\ref{eq:3}), (\ref{eq:4}) and (\ref{eq:5}), it follows that
\begin{equation*}
 \frac{\Var(X_c)}{\EE^2(X_c)} \leq \frac{1}{\EE(X_c)}
 +O\left(\frac{(\log n)^3}{n}\right)
 +\binom{n}{c}\frac{cg(1)}{\EE^2(X_c)}.
\end{equation*}
Since, by Lemma \ref{lem:1}, $\EE(X_c)\to \infty$ as $n\to \infty$, we will have $\Var(X_c)/\EE^2(X_c)=o(1)$ if
the last term in the right-hand side of the above inequality tends to zero as $n\to \infty$. We have
\begin{eqnarray}
\binom{n}{c}\frac{cg(1)}{\EE^2(X_c)}
&=& \frac{1}{\EE^2(X_c)}\times\frac{(n)_cc!n^{c-1}}{(c-1)!c^{2c-2}}
    \exp\big[(n-2c)\left((1-p)^{2c-1}-2(1-p)^{c}\right)\big]\nonumber\\
&=& \frac{c^3n^{c-1}}{(n)_c(1-(1-p)^{c})^{2n-2c}}
    \exp\big[(n-2c)\left((1-p)^{2c-1}-2(1-p)^{c}\right)\big]\nonumber\\
&=& \frac{c^3n^{c-1}}{(n)_c}
    \exp\big[(n-2c)\left((1-p)^{2c-1}-2(1-p)^{c}\right)-(2n-2c)\log(1-(1-p)^c)\big]\nonumber\\
&=& \frac{c^3n^{c-1}}{(n)_c}
    \exp\left[ O\left(\frac{(\log n)^4}{n}\right)\right].\nonumber
\end{eqnarray}
Since $c=o(\sqrt{n})$, we have $(n)_c=(1+o(1))n^c$. Thus,
\begin{equation*}
\binom{n}{c}\frac{cg(1)}{\EE^2(X_c)} = O\left(\frac{(\log n)^3}{n}\right)=o(1).
\end{equation*}
By Chebychev's inequality,
\begin{equation*}
 \PP[X_c=0]\leq \frac{\Var(X_c)}{\EE^2(X_c)}\to 0,
\end{equation*}
which completes the proof of the theorem.
\end{proof}

\section{Approximability and Fixed Parameter Tractability}\label{aprox}

We assume familiarity with the complexity classes NPO and PO
which are optimisation analogues of NP and P.
A minimisation problem in NPO is said to be \emph{approximable} within a constant $r \geq 1$
if there exists an algorithm $A$ which, for every instance $I$, outputs a solution of measure
$A(I)$ such that $A(I)/\opt(I) \leq r$, where $\opt(I)$ stands for the measure of an optimal solution.
An NPO problem is in the class APX if it is approximable within \emph{some} constant
factor $r \geq 1$.
An NPO problem is in the class PTAS if it is approximable within $r$ for \emph{every}
constant factor $r > 1$.
An APX-hard problem cannot be in PTAS unless P = NP.
We use two types of reductions, L-reductions to prove APX-hardness,
and PTAS-reductions to demonstrate inclusion in PTAS.
 In the Appendix we give a slightly more formal introduction and a description
 of reduction methods related to approximability.
For more on these issues we refer to
 Ausiello et al.~\cite{Ausiello:etal:CA} and Crescenzi~\cite{Crescenzi97}.

A problem is said to be \emph{fixed parameter tractable} (FPT) with parameter $k \in \mathbb{N}$
if it has an algorithm that runs in time
$f(k) \left|I\right|^{\Ordo(1)}$ for any instance $(I,k)$,
 where $f$ is an arbitrary function that depends only on $k$.

In this section, we study the complexity of approximating and solving TDS
conditioned on various restrictions on the input graphs and on the number
of colours.
First, we show that TDS is equivalent to MDS (Minimum Dominating Set)
 under L-reductions.
In particular, this implies that the general problem lies outside APX.
We then attempt to restrict the input graphs and observe that if MDS
is in APX on some family of graphs, then so is TDS.
However, there is also an immediate lower bound: TDS on
any family of graphs that contains all paths is APX-hard.
We proceed by adding an upper bound on the number of colours.
We see that if MDS is in PTAS for some family of graphs with bounded degree, then
so is TDS when restricted to $n^{1-\epsilon}$ colours for some $\epsilon > 0$.
Finally, we show that TDS on interval graphs is FPT with the parameter being
the number of colours and that the problem is in PO when the number of colours is logarithmic.

\begin{proposition}\label{tds-mds}
  TDS is equivalent to MDS under L-reductions.
  It is approximable within $\ln n + \Theta(1)$ but
  NP-hard to approximate within $(1-\epsilon) \ln n$.
\end{proposition}

\begin{proof}
  MDS is clearly a special case of TDS.
  For the opposite direction,
  we reduce an instance of TDS to an instance $I$ of the Set Cover problem
  which is known to be equivalent to MDS under L-reductions~\cite{Kannphd92}.
  In the Set Cover problem, we are given a ground set $U$ and a collection of
subsets $F_i \subseteq U$ such that $\bigcup_i F_i = U$.
The goal is to cover $U$ with the smallest possible number of sets $F_i$.
  Our reduction goes as follows.
  Given a vertex-coloured graph $G^c=(V^c,E)$, with the set of colours $C$,
   the ground set of $I$ is $U=V^c\cup C$.
  Each vertex $v$ of $V$ gives rise to a set $F_v=N[v]\cup \{c(v)\}$, a subset of $U$.
  Every solution to $I$ must cover every vertex $v \in V$ either by including a set
  that corresponds to $v$ or by including a set that corresponds to a neighbour of $v$.
  Furthermore, every solution to $I$ must include at least one vertex of every colour in $C$.
  It follows that every set cover can be translated back to a tropical dominating set of the same size.
  This shows that our reduction is an L-reduction.

  The approximation guarantee follows from that of the standard greedy algorithm for Set Cover.
  The lower bound follows from the NP-hardness reduction to Set~Cover in~\cite{DinSte2014}
  in which the constructed Set Cover instances contain $o(N)$ sets, where $N$ is the
  size of the ground set.
\end{proof}

When the input graphs are restricted to some family of graphs, then membership in APX
for MDS carries over to TDS.

\begin{lemma}\label{lem:inapx}
  Let $\mathcal{G}$ be a family of graphs.
  If MDS restricted to $\mathcal{G}$ is in APX, then TDS restricted to $\mathcal{G}$ is in APX.
\end{lemma}

\begin{proof}
  Assume that MDS restricted to $\mathcal{G}$ is approximable within
  $r$ for some $r \geq 1$.
  Let $G^c$ be an instance of TDS.
  We can find a dominating set of the uncoloured graph $G$ of size at most $r \gamma(G)$
  in polynomial time, and then add one vertex of each colour that is not yet present in the
  dominating set. This set is of size at most $r\gamma(G)+c-1$.
  The size of an optimal solution of $G^c$ is at least $\gamma(G)$ and at least $c$.
  Hence, the computed set will be at most $r+1$ times the size of the optimal solution of $G^c$.
\end{proof}

For $\Delta \geq 2$, let $\Delta$-TDS denote the problem of minimising
 a tropical dominating set on graphs of degree bounded by $\Delta$.
The problem MDS is in APX for bounded-degree graphs,
hence $\Delta$-TDS is in APX by Lemma~\ref{lem:inapx}.
The same lemma also implies that TDS restricted to paths is in APX.
Next, we give explicit approximation ratios for these problems.

\begin{proposition}\label{prop:path-approx}
  TDS restricted to paths can be approximated within $5/3$.
\end{proposition}

\begin{proof}
  Let $P^c = v_1, v_2, \dots, v_n$ be a vertex-coloured path. For $i = 1, 2, 3$ let
  $\sigma_i = \{ v_j \mid j \equiv i~($mod $3), \ 1\le j\le n\}$.
  Select any subset $\sigma'_i$ of $V$ that contains precisely one vertex
   of each colour missing from $\sigma_i$.
  Let $S_i = \sigma_i \cup \sigma'_i$.
By definition, $S_i$ is a tropical set.

%

Taking into account that each colour must appear in a tropical dominating set,
 moreover any vertex can dominate at most two others, we see the following
 easy lower bounds:
 \begin{eqnarray}
   n & \le & 3 \tdn (P^c) , \nonumber \\
   2c & \le & 2 \tdn (P^c) , \nonumber \\
  \frac{1}{5} (n+2c) & \le & \tdn (P^c) . \nonumber
 \end{eqnarray}

 Suppose for the moment that each of $S_1,S_2,S_3$ dominates $G^c$.
Then, since each colour occurs in at most two of the $\sigma'_i$, we have
 $|S_1|+|S_2|+|S_3| \le n+2c$ and therefore
 $$
   \tdn (P^c) \le \min (|S_1|,|S_2|,|S_3|) \le \frac{1}{3} (n+2c) .
 $$
Comparing the lower and upper bounds, we obtain that the smallest set $S_i$ provides
 a 5/3-approximation.
It is also clear that this solution can be constructed in linear time.

The little technical problem here is that the set $S_i$ does not dominate
 vertex $v_1$ if $i=3$, and it does not dominate $v_n$ if $i\equiv n-2$ (mod 3).
We can overcome this inconvenience as follows.

The set $S_3$ surely will dominate $v_1$ if we extend $S_3$ with
 either of $v_1$ and $v_2$.
This means no extra element if we have the option to select e.g.\ $v_1$
 into $\sigma'_3$.
  We cannot do this only if $c(v_1)$ is already present in $\sigma_3$.
But then this colour is common in $\sigma_1$ and $\sigma_3$; that is, although we
 take an extra element for $S_3$, we can subtract 1 from the term $2c$ when
 estimating $|\sigma'_1|+|\sigma'_2|+|\sigma'_3|$.
The same principle applies to the colour of $v_n$, too.

Even this improved computation fails by 1 when $n\equiv 2$ (mod 3) and
 $c(v_1)=c(v_n)$, as we can then write just $2c-1$ instead of $2c-2$ for
 $|\sigma'_1|+|\sigma'_2|+|\sigma'_3|$.
Now, instead of taking the vertex pair $\{v_1,v_n\}$ into $S_3$,
 we complete $S_3$ with $v_2$ and $v_n$.
This yields the required improvement to $2c-2$, unless $c(v_2)$, too,
 is present in $\sigma_3$.
But then $c(v_2)$ is a common colour of $\sigma_2$ and $\sigma_3$, while
 $c(v_1)$ is a common colour of $\sigma_1$ and $\sigma_3$.
Thus $|\sigma'_1|+|\sigma'_2|+|\sigma'_3| \le 2c-2$, and
 $|S_1|+|S_2|+|S_3| \le n+2c$ holds also in this case.
\end{proof}

\begin{remark}
 In an analogous way --- which does not even need the particular
  discussion of unfavourable cases --- one can prove that the
  square grid $P_n\Box P_n$ admits an asymptotic 9/5-approximation.
 (This extends also to $P_n\Box P_m$ where $m=m(n)$ tends to
  infinity as $n$ gets large.)
 A more precise estimate on grids, however, may require a
  careful and tedious analysis.
\end{remark}

\begin{proposition}
  $\Delta$-TDS is approximable within $\ln (\Delta+2) + \frac{1}{2}$.
  Moreover, there are absolute constants $C > 0$ and $\Delta_0 \geq 3$ such that
  for every $\Delta \geq \Delta_0$,
  it is NP-hard to approximate $\Delta$-TDS within $\ln \Delta - C \ln \ln \Delta$.
\end{proposition}

\begin{proof}
  The second assertion follows from
   \cite[Theorem~3]{ChCh08}.
  For the first part, we apply reduction from Set Cover,
   similarly as in the proof of Proposition \ref{tds-mds}.
So, for $G^c=(V^c,E)$ we define $U=V^c\cup C$ and consider
 the sets $F_v=N[v]\cup \{c(v)\}$ for the vertices $v\in V^c$.
Every set cover in this set system corresponds to a tropical
 dominating set in $G^c$.
  Moreover, the Set Cover problem is approximable within
   $\sum_{i=1}^k \frac {1}{i} - \frac{1}{2} < \ln k + \frac{1}{2}$~\cite{DuhFurer97},
  where $k$ is an upper bound on the cardinality of any set of $I$.
  In our case, we have $k = \Delta + 2$ since $\left|N(v)\right| \leq \Delta$ for all $v$.
  Hence, TDS is approximable within $\ln (\Delta+2) + \frac{1}{2}$.
\end{proof}

We now show that TDS for paths is APX-complete.

\begin{theorem}\label{thm:apxc}
  TDS restricted to paths is \cc{APX}-hard.
\end{theorem}

\begin{proof}
  We apply an L-reduction from the Vertex Cover problem (VC): Given a graph $G = (V,E)$,
  find a set of vertices $S \subseteq V$ of minimum cardinality
   such that, for every edge
  $uv \in E$, at least one of $u \in S$ and $v \in S$ holds.
  We write 3-VC for the vertex cover problem restricted to graphs
  of maximum degree three (subcubic graphs).
  The problem 3-VC is known to be APX-complete~\cite{AK00}.
  For a graph $G$, we write $\opt_{VC}(G)$ for the minimum
  size of a vertex cover of $G$.

   Let $G = (V,E)$ be a non-empty instance of 3-VC, with
   $V=\{v_1,\dots,v_n\}$ and $E=\{e_1,\dots,e_m\}$.
Assume that $G$ has no isolated vertices.
  The reduction sends $G$ to an instance $\phi(G)$ of TDS which will have $m+n+1$ colours:
  $B$ (for black),
  $E_i$ with $1 \leq i \leq m$ (for the $i$th edge), and
  $S_j$ with $1 \leq j \leq n$ (for the $j$th vertex).
 The path has $9n+3$ vertices altogether, starting with three black
  vertices of  Figure~\ref{fig:gadgets}$(a)$, we call this triplet $V_0$.
 Afterwards blocks of 6 and 3 vertices alternate, we call the latter
   $V_1, \dots, V_n$, representing the vertices of $G$.
 Each $V_j$ (other than $V_0$) is coloured as shown in Figure~\ref{fig:gadgets}$(c)$.
 Assuming that $v_j$ ($1\le j\le n$) is incident to the edges
  $e_{j_1}$, $e_{j_2}$, and $e_{j_3}$, the two parts
    $V_{j-1}$ and $V_j$ are joined by a path
    representing these three incidences, and
   coloured as in Figure~\ref{fig:gadgets}$(b)$.
 If $v_j$ has degree less than 3, then the vertex in place of $E_{j_3}$ is black;
  and if $d(v_j)=1$, then also $E_{j_2}$ is black.

\FloatBarrier
\begin{figure}[h]
	\centering
	\includegraphics[scale=.5]{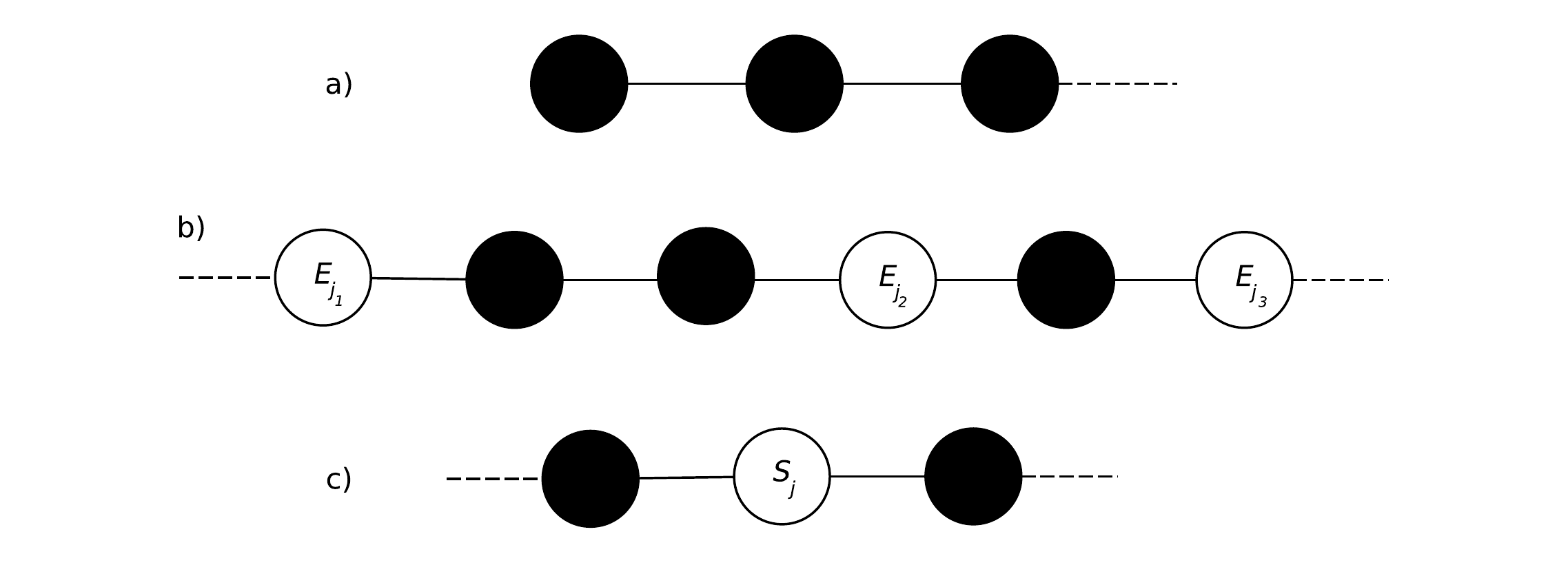}
	\caption{Gadgets for the reduction of Theorem~\ref{thm:apxc}}
	\label{fig:gadgets}
\end{figure}
\FloatBarrier

  Let $\sigma \subseteq V$ be an arbitrary solution to $\phi(G)$.
  First, we construct a solution $\sigma'$ from $\sigma$ with more structure,
  and with a measure at most that of $\sigma$.
  For every $j$, $\sigma$ contains the vertex coloured $S_j$.
  Let $\sigma'$ contain these as well.
  At least one of the first two vertices coloured $B$ must also be in $\sigma$.
  Let $\sigma'$ contain the second vertex coloured $B$.
Now, if any $V_j$ ($0\le j\le n$) has a further (first or third) vertex which
 is an element of $\sigma$, then we can replace it with its predecessor
 or successor, achieving that they dominate more vertices in the path.
This modification does not lose any colour because the first and third
 vertices of any $V_j$ are black, and B is already represented in $\sigma\cap V_0$.

 Now we turn to the 6-element blocks connecting a $V_{j-1}$ with $V_j$.
Since the third vertex of $V_{j-1}$ and the first vertex of $V_j$ are
 surely not in the modified $\sigma$, which still dominates the path,
 it has to contain at least two vertices of the 6-element block.
And if it contains only two, then those necessarily are the second and fifth,
 both being black.
Should this be the case, we keep them in $\sigma'$.
 Otherwise, if the modified
 $\sigma$ contains more than two vertices of the 6-element block,
 then let $\sigma'$ contain precisely $E_{j_1}$, $E_{j_2}$, and $E_{j_3}$.
  Since $\sigma$ is a tropical dominating set, the same holds for $\sigma'$.
  It is also clear that $\left|\sigma'\right| \leq \left|\sigma\right|$.

  Next, we create a solution $\psi(G,\sigma)$ to the vertex cover problem on $G$, using $\sigma'$.
  Let $v_j \in \psi(G,\sigma)$ if and only if
   $\{E_{j_1},E_{j_2},E_{j_3}\} \subseteq \sigma'$.
  Then,
  $
  \left|\psi(G,\sigma)\right| = \left|\sigma'\right|-1-3n \leq \left|\sigma\right|-1-3n,
  $
  and when $\sigma$ is optimal, we have the equality
  $
  \opt_{VC}(G) = \tdn(\phi(G)) -1-3n.
  $
  Therefore,
  \begin{equation}\label{eq:beta}
  \left|\psi(G,\sigma)\right|-\opt_{VC}(G) \leq
  \left|\sigma\right|-\tdn(\phi(G)).
  \end{equation}

  We may assume that $G$ does not contain any isolated vertices.
  Under this assumption, we prove the lower bound
  $
  \opt_{VC}(G) \geq n/4
  $
  by induction, as follows:
  The bound clearly holds for an empty graph.
  Suppose that the bound holds for all graphs without isolated vertices
  with fewer than $n$ vertices.  
  Let $\sigma^\ast$ be a minimal vertex cover of $G$ and let $v \in V \setminus \sigma^\ast$.
  Then, all of $v$'s neighbours are in $\sigma^\ast$.
  Let $G'$ be the graph $G$ with $N[v]$ removed as well as any isolated vertices resulting
  from this removal.
  Let $n'$ be the number of vertices in $G'$.
  If $v$ has $1 \leq n_v \leq 3$ neighbours,
  then $0 \leq n_i \leq 2n_v$ vertices become isolated when $N[v]$ is removed,
  so
  $\opt_{VC}(G) = n_v+\opt_{VC}(G') \geq n_v+ n'/4 = n_v+(n-1-n_v-n_i)/4 \geq 
  n_v+(n-1-3n_v)/4 \geq n/4$.

  This allows us to upper-bound the optimum of $\phi(G)$:
  \begin{align}
    \tdn(\phi(G))
    & = \opt_{VC}(G)+1+3n \notag \\
    & \leq \opt_{VC}(G)+1+12 \cdot \opt_{VC}(G) \leq 14 \cdot \opt_{VC}(G). \label{eq:alpha}
  \end{align}
  It follows from (\ref{eq:beta}) and (\ref{eq:alpha}) that $\phi$ and $\psi$ constitute an L-reduction.
\end{proof}


\begin{corollary}\label{cor:neps}
  Fix $0 < \epsilon \leq 1$, and
  let $\mathcal{P}$ be the family of all vertex-coloured paths
  with at most $n^{\epsilon}$ colours, where $n$ is the number of vertices.
  Then TDS restricted to $\mathcal{P}$ is NP-hard.
\end{corollary}

\begin{proof}
  We reduce from TDS on paths with an unrestricted number of colours which is NP-hard by Theorem~\ref{thm:apxc}.
  Let $P^c$ be a vertex-coloured path on $n$ vertices with $c \leq n$ colours.
  Let $Q^{c'}$ be the instance obtained by adding a path $v_1, v_2, \dots, v_N$ with 
  $N = \lceil (n+2)^{1/\epsilon} \rceil$ vertices to the end of $P^c$
  (this is a polynomial-time reduction for any fixed constant $\epsilon > 0$).
  Let $A$ and $B$ be two new colours.
  In the added path $v_1, v_2, \dots, v_N$, let $v_2$ have colour $A$ and all the other
  vertices have colour $B$.
  The instance $Q^{c'}$ has $n' = n+N$ vertices and 
  $c' = c+2 \leq n+2 \leq N^{\epsilon} \leq (n')^{\epsilon}$ colours, so $Q^{c'} \in \mathcal{P}$.

  Given a minimum tropical dominating set $\sigma$ of $Q^{c'}$, we see that $v_2$ must be in $\sigma$
  to account for the colour $A$.
  We may further assume that $v_1$ is not in $\sigma$. If it were, then we could modify $\sigma$ by
  removing $v_1$ and adding the last vertex of $P^c$ instead.
  It is now clear that
  taking $\sigma$ restricted to $\{v_1, v_2, \dots, v_N\}$ together with a tropical dominating set of $P^c$
  yields a tropical dominating set of $Q^{c'}$ and that
  $\sigma$ restricted to $P^c$ is a tropical dominating set of $P^c$.
  Hence, $\sigma$ restricted to $P^c$ is a minimum tropical dominating set of $P^c$.
\end{proof}

We have seen that restricting the input to any graph family that contains at least the paths
can take us into APX but not further.
To find more tractable restrictions, we now introduce an additional
restriction on the number of colours.
The following lemma says that if the domination number grows
asymptotically faster than the number of colours, then we can lift
PTAS-inclusion of MDS to TDS.

\begin{lemma}
\label{lem:ptas-red}
  Let $\mathcal{G}$ be a family of vertex-coloured graphs.
  Assume that there exists a computable function $f \colon \Q \cap (0,\infty) \to \N$
  such that for every $r > 0$, $\gamma(G) > c/r$ whenever $G^c \in \mathcal{G}$ and $n(G^c) \geq f(r)$.
  Then,
  TDS restricted to $\mathcal{G}$ PTAS-reduces to
  MDS restricted to $\mathcal{G}$.
\end{lemma}

\begin{proof}
 To design a polynomial-time $(1+\varepsilon)$-approximation for any rational $\varepsilon > 0$,
  we pick $r=\varepsilon/2$; hence let $n_0 = f(\varepsilon/2)$.
  Let $G^c \in \mathcal{G}$ be a vertex-coloured graph.
  The reduction sends $G^c$ to $\phi(G^c) = G$,
  the instance of MDS obtained from $G^c$ by simply forgetting the colours.
  Let $\sigma$ be any dominating set in $G$.
  Assuming that $\sigma$ is a good approximation to $\gamma(G)$,
  we need to compute a good approximation $\psi(G^c, \sigma)$ to $\tdn(G^c)$.
  If $n(G^c) < n_0$, then we let $\psi(G^c,\sigma)$ be an optimal tropical dominating set of $G^c$.
  Otherwise, let $\psi(G^c,\sigma)$ be $\sigma$ plus a vertex for each remaining non-covered colour.
  Since $n_0$ depends on $\varepsilon$ but not on $G^c$ or $\sigma$,
  it follows that $\psi$ can be computed in time that is polynomial in $|V(G^c)|$ and $|\sigma|$.

  We claim that $\phi$ and $\psi$ provide a PTAS-reduction.
  This is clear if $n(G^c) < n_0$ since $\psi$ then computes an optimal solution to $G^c$.
  Otherwise, assume that $n(G^c) \geq n_0$ and that $\left|\sigma\right|/\gamma(G) \leq 1+\varepsilon/2$,
  i.e., $\sigma$ is a good approximation.
  Then,
  \[
  \frac{\left|\psi(G^c, \sigma)\right|}{\tdn(G^c)}
  \leq \frac{\left|\sigma\right|+c}{\gamma(G)}
  \leq \frac{2+\varepsilon}{2} + \frac{c}{\gamma(G)} < 1+\varepsilon,
  \]
  where the last inequality follows from $n(G) \geq n_0$ and the definition of $f$.
\end{proof}

\begin{example}
  The problem MDS is in \cc{PTAS} for planar graphs~\cite{Baker94},
  but \cc{NP}-hard even for planar subcubic graphs~\cite{GareyJohnson}.
  Let $\mathcal{G}$ be the family of planar graphs of maximum degree $\Delta$,
  for any fixed $\Delta \geq 3$,
  and with a number of colours $c < n^{1-\epsilon}$ for some fixed $\epsilon > 0$.
  Let $f(r) = \lceil ( \frac{\Delta+1}{r} )^{1/\epsilon} \rceil$ and note that
  $\gamma(G) \geq n/(\Delta+1) > cn^\epsilon/(\Delta+1) \geq c f(r)^\epsilon/(\Delta+1) \geq c/r$
  whenever $n \geq f(r)$.
  It then follows from Lemma~\ref{lem:ptas-red} that TDS is in \cc{PTAS}
  when restricted to planar graphs of fixed maximum degree.
\end{example}

\begin{example}
As a second example, we observe how the complexity of TDS on a path varies when
we restrict the number of colours.
For an arbitrary number of colours, it is APX-complete by Lemma~\ref{lem:inapx} and
Theorem~\ref{thm:apxc}.
If the number of colours is $\Ordo(n^{1-\epsilon})$ for some $\epsilon > 0$, then
it is in PTAS by Lemma~\ref{lem:ptas-red}, but NP-hard by Corollary~\ref{cor:neps}.
Finally, if the number of colours is $\Ordo(\log n)$, then it can be shown to be in PO by a simple
dynamic programming algorithm.
\end{example}

In the rest of this section, we look at the restriction where we consider the number of colours
as a fixed parameter.
We prove the following result.

\begin{theorem}\label{thm:interval-fpt}
  There is an algorithm for TDS restricted to interval graphs that runs in time $\Ordo(2^cn^2)$.
\end{theorem}

This shows that TDS for interval graphs is FPT and, furthermore, that if
$c = \Ordo(\log n)$, then TDS is in PO.

Let $G^c$ be a vertex-coloured interval graph with vertex set $V = \{1, \dots, n\}$ and colour set $C$, and fix
some interval representation $I_i = [l_i,r_i]$ for each vertex $1 \leq i \leq n$.
Assume that the vertices are ordered non-decreasingly with respect to $r_i$.
For $a, b \in V$, we use (closed) intervals $[a,b] = \{ i \in V \mid a \leq i \leq b \}$ to denote subsets of vertices with respect to this order.

Define an \emph{$i$-prefix dominating set} as a subset $U \subseteq V$ of vertices that contains $i$ and dominates $[1,i]$ in $G^c$.
We say that $U$ is \emph{proper} if, for every $i, j \in U$,
we have neither $I_i \subseteq I_j$ nor $I_j \subseteq I_i$.

Let $f \colon \mathcal{P}(C) \times \left[0, n \right] \to \N \cup \{\infty\}$ be the function defined so that,
given a subset $S \subseteq C$ of colours and a vertex $i \in V$,
$f(S,i)$ is the least number of vertices in a proper $i$-prefix dominating set
that covers precisely the colours in $S$, or $\infty$ if there is no such set.
The value of $f(S,0)$ is defined to be $0$ when $S = \emptyset$ and $\infty$ otherwise.
Our proof is based on a recursive definition of $f$ (Lemma~\ref{lem:recursive})
and the fact that $f$ determines $\tdn$ (Lemma~\ref{lem:interval}).
First, we need a technical lemma.

\begin{lemma}
  \label{lem:rightmost}
  Let $U \subseteq V$ and let $i$ be the largest element in $U$.
  If $U$ is $i$-prefix dominating,
  then it dominates precisely the same vertices as $[1,i]$.
  In particular, $U$ dominates $G$ if and only if $[1,i]$ does.
\end{lemma}

\begin{proof}
  Assume to the contrary that there is a $j \in [1,i] - U$
  that dominates some $k > i$, and that $k$ is not dominated by $U$.
  This means that $j$ is connected to $k$ in $G$, so $l_k \leq r_j$.
  But then we have $l_k \leq r_j \leq r_i \leq r_k$, so
  $[l_i,r_i] \cap [l_k,r_k] \neq \emptyset$, hence $i \in U$ dominates $k$,
  a contradiction.
\end{proof}

\begin{lemma}
\label{lem:interval}
  For every interval graph $G^c$, we have
  \[
    \tdn(G^c) = \min \{ f(S,i) + \left|C - S\right| \mid S \subseteq C, i \in V, [1,i] \textrm{ dominates } G^c \}.
  \]
\end{lemma}

\begin{proof}
  $f(S,i)$ is the size of some set $U \subseteq V$ that covers the colours $S$ and that, by Lemma~\ref{lem:rightmost}, dominates $G^c$.
  We obtain a tropical dominating set by adding a vertex of each missing colour in $C - S$.
  Therefore, each expression $f(S,i) + \left|C - S\right|$ on the right-hand side corresponds to the size of a tropical dominating set, so $\tdn(G^c)$ is at most the minimum of these.

  For the opposite inequality,
  let $U$ be a minimum tropical dominating set of $G^c$.
  Remove from $U$ all vertices $i$ for which there is some $j \in U$ with
  $I_i \subseteq I_j$, and call the resulting set $U'$.
  By construction $U'$ still dominates $G^c$.
  Let $S$ be the set of colours covered by $U'$.
  Then $U'$ is a minimum set with these properties,
  so by the definition of $f$,
  $\left|U'\right| = f(S,i)$, where $i$ is the greatest element in $U'$.
  Since $U' \subseteq [1,i]$, it follows that $[1,i]$ dominates $G^c$.
  Therefore, the right-hand side is at most $f(S,i) + \left|C - S\right| = \left|U'\right| + \left|C - S\right| \leq \left|U\right| = \tdn(G^c)$.
\end{proof}

The following lemma gives a recursive definition of the function $f$ that permits us
to compute it efficiently when the number of colours in $C$ grows at most logarithmically.

\begin{lemma}
  \label{lem:recursive}
  For every interval graph $G^c$, the function $f$ satisfies the following recursion:
  \[
  f(S,0) = \begin{cases}
  0 & \textrm{ if $S = \emptyset$,}\\
  \infty & \textrm{ otherwise;}
  \end{cases}
  \]
  \[
  f(S,i) = 1 + \min \{ f(S',j) \mid S' \cup \{ c(i) \} = S, j \in P_i \}, \qquad \textrm{ for $i \in V,$}
  \]
  where $j \in P_i$ if and only if either $j = 0$ and $\{i\}$ is $i$-prefix dominating, or
  $j \in V$, $j < i$,  $\left[1,j\right] \cup \{i\}$ is $i$-prefix dominating, and $I_i \not\subseteq I_j$, $I_j \not\subseteq I_i$.
\end{lemma}

\begin{proof}
  The proof is by induction on $i$.
  The base case $i = 0$ holds by definition.
  Assume that the lemma holds for all $0 \leq i \leq k-1$ and all $S \subseteq C$.

  Let $U$ be a minimum proper $k$-prefix dominating set that covers precisely the colours in $S$.
  We want to show that $\left|U\right| = f(S,k)$.
  If $U = \{k\}$, then $S = \{c(k)\}$, and it follows immediately that $f(S,k) = 1$.
  Otherwise, $U - \{k\}$ is non-empty.
  Let $j < k$ be the greatest vertex in $U - \{k\}$.
  Assume that $U - \{k\}$ is not $j$-prefix dominating.
  Then, there is some $i < j$ that is not dominated by $j$ but that is
  dominated by $k$, hence $l(k) \leq r(i) < l(j)$.
  Therefore $I_j \subseteq I_k$, so $U$ is not proper, a contradiction.
  Hence, $U - \{k\}$ is a proper $j$-prefix dominating set.
  By induction, $\left|U - \{k\}\right| \geq \min \{ f(S',j) \mid S' \cup \{c(k)\} = S \}$.
  This shows the inequality $\left|U\right| \geq f(S,k)$.

  For the opposite inequality, it suffices to show that if $\left[1,j\right] \cup \{k\}$ is $k$-prefix dominating,
  $U'$ is any proper $j$-prefix dominating set, and $I_k \not\subseteq I_j, I_j \not\subseteq I_k$, then
  $U' \cup \{k\}$ is a proper $k$-prefix dominating set.
  It follows from Lemma~\ref{lem:rightmost} that $U' \cup \{k\}$ is $k$-prefix dominating.
  Since $I_k \not\subseteq I_j$, we must have $r_i \leq r_j < r_k$ for all $i < j$,
  hence $I_k \not\subseteq I_i$.
  Assume that $I_i \subseteq I_k$ for some $i < j$.
  Then, since $I_j \not\subseteq I_k$, we have $l_j < l_k \leq l_i \leq r_i \leq r_j$,
  which contradicts $U'$ being proper.
  It follows that $U' \cup \{k\}$ is proper.
\end{proof}

\noindent
{\it Proof of Theorem~\ref{thm:interval-fpt}.}
The sets $P_i$ for $i \in V$ in Lemma~\ref{lem:recursive}
can be computed in time $\Ordo(n^2)$ as follows.
Let $a_i \in V$ be the least vertex such that $i$ dominates $[a_i,i]$,
and
let $b_j \in V$ be the least vertex such that $[1,j]$ does not dominate $b_j$,
or $\infty$ if $[1,j]$ dominates $G$.
Note that $i$ does not dominate any vertex strictly smaller than $a_i$
since the vertices are ordered non-decreasingly with respect to the
right endpoints of their intervals.
Therefore, $P_i = \{ j < i \mid a_i \leq b_j, I_i \not\subseteq I_j, I_j \not\subseteq I_i \}$.
The vectors $a_i$ and $b_j$
are straightforward to compute in time $\Ordo(n^2)$,
hence $P_i$ can be computed in time $\Ordo(n^2)$ using this alternative definition.

When $P_i$ is computed for all $i \in V$,
the recursive definition of $f$ in Lemma~\ref{lem:recursive} can be used
to compute all values of $f$ in time $\Ordo(2^c n^2)$,
and it can easily be modified to compute, for each $S$ and $i$,
some specific $i$-prefix dominating set of size $f(S,i)$,
also in time $\Ordo(2^c n^2)$.
Therefore, by Lemma~\ref{lem:interval}, one can
find a minimum tropical dominating set in time $\Ordo(2^c n^2)$.
 \hfill $\Box$

\paragraph{Acknowledgements.}

The research of Cs.\ Bujt\'as and Zs.\ Tuza was supported in part by
 the European Union and Hungary, co-financed by the European Social Fund
 through the project T\'AMOP-4.2.2.C-11/1/KONV-2012-0004.

\bibliography{ref}


\end{document}